\documentclass{article}

\usepackage{arxiv}

\usepackage[utf8]{inputenc} 
\usepackage[T1]{fontenc}    
\usepackage{hyperref}       
\usepackage{url}            
\usepackage{booktabs}       
\usepackage{amsfonts}       
\usepackage{nicefrac}       
\usepackage{microtype}      
\usepackage{lipsum}
\usepackage{graphicx}
\usepackage{microtype}      
\usepackage{lipsum}
\usepackage{graphicx}
\graphicspath{{media/}}   
\usepackage{cite}
\usepackage{amsmath,amssymb,amsfonts}
\usepackage{amsthm}
\usepackage{algorithmic}
\usepackage{graphicx}
\usepackage{textcomp}
\usepackage{xcolor}

\newtheorem{prop}{Proposition}
\usepackage{url}

\graphicspath{ {./images/} }

\title{Re-Evaluating Privacy in Centralized and Decentralized Learning: An Information-Theoretical and Empirical Study}
\author{
  Changlong Ji \\
  Telecom Sudparis \\
  Institut Polytechnique de Paris \\
  France\\
  \texttt{changlong.ji@telecom-sudparis.eu} \\
  \And
  Stephane Maag \\
  Telecom Sudparis \\
  Institut Polytechnique de Paris \\
  France\\
  \texttt{stephane.maag@telecom-sudparis.eu} \\
  \And
  \hspace{-3em}Richard Heusdens \\
   \hspace{-3em}Netherlands Defence Academy \\
  \hspace{-3em} Delft University of Technology \\
   \hspace{-3em}Netherlands\\
  \hspace{-3em} \texttt{r.heusdens@tudelft.nl} \\
   \And
  \hspace{2em} Qiongxiu Li \\
  \hspace{2em} Aalborg University \\
  \hspace{2em} Denmark\\
  \hspace{2em} \texttt{qili@es.aau.dk} \\
}

\begin{document}
\maketitle

\begin{abstract}

Decentralized Federated Learning (DFL) has garnered attention for its robustness and scalability compared to Centralized Federated Learning (CFL). While DFL is commonly believed to offer privacy advantages due to the decentralized control of sensitive data, recent work by Pasquini et, al. challenges this view, demonstrating that DFL does not inherently improve privacy against empirical attacks under certain assumptions. For investigating fully this issue, a formal theoretical framework is required. Our study offers a novel perspective by conducting a rigorous information-theoretical analysis of privacy leakage in FL using mutual information. We further investigate the effectiveness of privacy-enhancing techniques like Secure Aggregation (SA) in both CFL and DFL. Our simulations and real-world experiments show that DFL generally offers stronger privacy preservation than CFL in practical scenarios where a fully trusted server is not available. We address discrepancies in previous research by highlighting limitations in their assumptions about graph topology and privacy attacks, which inadequately capture information leakage in FL.

\end{abstract}


\section{Introduction}
Federated Learning (FL) enables collaborative training without sharing raw data \cite{McMahan2016CommunicationEfficientLO}, using either centralized or decentralized topologies\cite{MartnezBeltrn2022DecentralizedFL,Li2019FederatedLC,Shi2023ImprovingTM,yu2024provable,li2022privacy}. Centralized FL (CFL), which relies on a central server for node coordination, is widely adopted but suffers from high bandwidth demands \cite{Hu2019DecentralizedFL, Daily2018GossipGraDSD} and introduces risks such as single points of failure \cite{MartnezBeltrn2022DecentralizedFL}. 
Although FL avoids direct sharing of raw data, it remains susceptible to privacy breaches  \cite{Geng2021TowardsGD,Boenisch2021WhenTC,Li2023TopologyDependentPB}. Such breaches can occur during the exchange of gradients or model weights, which can be exploited by attacks, such as membership inference \cite{Shokri2016MembershipIA}, where an adversary determines whether a specific data point was part of the training set, and gradient inversion which can reconstruct original data from shared gradients with high fidelity \cite{2019arXiv190608935Z,geiping2020inverting}. To mitigate these vulnerabilities, cryptographic techniques like Differential Privacy (DP) \cite{9714350} and SA \cite{2016arXiv161104482B} have been proposed. DP introduces noise into model updates to protect data, but often results in a trade-off between accuracy and privacy \cite{9413764}, whereas SA maintains accuracy by only sharing aggregated updates, though at the expense of increased communication overhead \cite{2020arXiv201205433C}.

Given that privacy concerns are a core motivation for FL \cite{McMahan2016CommunicationEfficientLO}, it is crucial to determine which graph topology, centralized or decentralized, is more effective for privacy preservation. However, there is ongoing debate on this complex question.  Some studies \cite{cheng2019towards,vogels2021relaysum} claim that the DFL is inherently more privacy-preserving than CFL, arguing that the sensitive information, such as private data and model parameters, is no longer concentrated in a central server.  Yet, these claims are primarily intuitive and lack both theoretical and empirical support. Pasquini et al. \cite{Pasquini2022OnTP} challenge this view, demonstrating through empirical attacks that DFL does not consistently provide better privacy protection than CFL and, in some cases, may increase vulnerability to privacy breaches. However, these findings are largely empirical and lack a rigorous theoretical foundation. The ongoing debate has motivated us to perform a thorough analysis. Our main contributions are: 

 \begin{itemize}
     \item We present a comprehensive information-theoretic analysis of privacy leakage in CFL and DFL, using mutual information as the primary metric.
     \item Our simulations and empirical privacy attacks reveal that DFL generally provides stronger privacy protection than CFL, particularly in scenarios without a fully trusted server, highlighting DFL’s effectiveness in mitigating privacy risks. 
     \item We addressed inconsistencies in prior research by identifying key limitations in their assumptions about graph topology and privacy attacks, which may not fully reflect the dynamics of information leakage in FL.
 \end{itemize}

\section{Preliminaries}

\subsection{Centralized FL}

    
    


In CFL, a centralized server is assumed to be connected to \(n\) clients, where \(\mathbf{w}^{(t)}\) denotes the model weights at iteration \(t\). The typical FL protocol, FedAvg \cite{McMahan2016CommunicationEfficientLO}, starts by initializing the global model weights \(\mathbf{w}^{(0)}\) at iteration \(t = 0\). Each client then trains the model on its local dataset for a fixed number of epochs, resulting in updated local model weights \(\mathbf{w}_i^{(t)}\). These updated weights are aggregated by the central server through averaging. The same idea applies to gradient sharing in the FedSGD \cite{McMahan2016CommunicationEfficientLO} protocol, where instead of sharing model weights, clients compute local gradients \(\mathbf{g}_i^{(t)}\), and the server aggregates these gradients and updates the global model as follows:
\begin{equation}
\mathbf{w}^{(t+1)} = \mathbf{w}^{(t)} - \frac{\eta}{n} \sum_{i=1}^{n} \mathbf{g}_i^{(t)},
\label{eq:fedsgd}
\end{equation}
where \(\eta > 0\) is the stepsize of the gradient update.

\subsection{Decentralized FL}
DFL works for a distributed network, which is often modeled as an undirected graph $\mathcal{G} = (\mathcal{V}, \mathcal{E})$, where $\mathcal{V} = \{1, 2, \ldots, n\}$ denotes the set of nodes, and $\mathcal{E} \subseteq \mathcal{V} \times \mathcal{V}$ represents the set of edges. It is important to note that node $i$ can only communicate with its neighboring node $j$ if and only if $(i, j) \in \mathcal{E}$. $m=|\mathcal{E}|$ is the number of edges. 

Protocols for DFL typically follows the same update procedure as the centralized approach, except at the model aggregation step. In this step, peer-to-peer distributed algorithms, such as linear iteration \cite{Olshevsky2009ConvergenceSI}, are employed to average local gradients \cite{Li2023TopologyDependentPB}. Another line of research seeks to address the learning problem through distributed optimization, combining both local model updates and model aggregation \cite{niwa2020edge,yu2024provable,Li2024eusipco}. In this paper, we focus on the former aggregation-based approach, as it is more widely adopted. Specifically, the gradient aggregation is performed iteratively by each node in conjunction with its neighbors. The gradient averaging step of FedSGD is replaced with peer-to-peer information exchange, which can be expressed as:
\begin{align}
  \mathbf{g}^{(t+1)} = \mathbf{A} \mathbf{g}^{(t)},
\end{align}
where \(\mathbf{g} = \left( \mathbf{g}_1^\top, \mathbf{g}_2^\top, \ldots, \mathbf{g}_n^\top \right)^\top\) is the vector of stacked local gradients of all nodes and $\mathbf{A}$ is a weight matrix that characterizes the graph's connectivity: $\left\{\mathbf{A}\in \mathbb{R}^{n \times n} \,|\, \mathbf{A}_{i j}=0 \text { if }(i,j) \notin \mathcal{E} \text { and } i \neq j\right\}.$ To ensure convergence to the average, the weight matrix $\mathbf{A}$ must meet the following conditions: 
(i) $\mathbf{1}^{\top}\mathbf{A}=\mathbf{1}^{\top}$, 
(ii) $\mathbf{A}\mathbf{1}=\mathbf{1}$, and 
(iii) $\rho\left(\mathbf{A}-\frac{\mathbf{1} \mathbf{1}^{\top}}{n}\right)<1$, 
where $\rho(\cdot)$ denotes the spectral radius \cite{xiao2004fast}. 
This protocol is known as D-FedSGD \cite{Sun2021DecentralizedFA}. 


\subsection{Secure aggregation}
SA \cite{2016arXiv161104482B} is a cryptographic protocol that can be used to enhance privacy in FL. It encrypts participants' model updates, allowing the server to aggregate these updates without decrypting them, thereby safeguarding the sensitive data of each participant. Many approaches have been proposed to securely aggregate private data in distributed networks such as secure multiparty computation approaches \cite{li2019privacyA,li2019privacyS,tjell2020privacy,gupta2017privacy} and subspace perturbation approaches \cite{Jane2020ICASSP,li2020privacy,li2023adaptive}.  It ensures that individual model updates, such as gradients, remain confidential even from the central server or neighboring nodes in DFL. In CFL, the server only receives the aggregated model updates, while in DFL, each node can only access the aggregated gradients of its neighboring nodes. 

\subsection{Adversary model}

We consider a passive adversary model, often referred to as "honest-but-curious," which is frequently employed in the study of distributed systems. Passive adversaries operate by colluding with multiple nodes within the network, which are referred to as adversary nodes or corrupt nodes. 

\section{Privacy comparison via information-theoretical analysis}



\subsection{Privacy analysis by mutual information}


\vspace{-0mm}
We conduct the first information-theoretical analysis of this problem, leveraging mutual information \cite{cover2012elements} as the metric and focusing on two network topologies in FL: centralized and decentralized. Both topologies are evaluated under two operational modes, with and without SA, leading to four distinct configurations: CFL w/o SA, CFL w/ SA, DFL w/o SA, and DFL w/ SA, where \texttt{w/} and \texttt{w/o} denote "with" and "without", respectively. We consider $x$ as a realization of a random variable, the corresponding random variable will be denoted by $X$ (corresponding capital). Let \( \mathcal{A} \) represent the total information that an adversary can observe. The mutual information \( I(\mathcal{A}; G_i) \) quantifies how much the adversary learns about the local gradient  \( G_i \) at honest node \( i \) based on the information contained in \( \mathcal{A} \). 
Assuming a passive adversary could be any node in the network, denoted as \(N_k\), and following the setup in \cite{Pasquini2022OnTP} with a single corrupt node, we calculate the average privacy leakage across the network (results can be easily generalized to multiple corrupt nodes). Denote \( I_\text{FL}^\text{mode} \) as  the average mutual information of honest nodes revealed to the adversary, defined as: 
\vspace{-0mm}
\begin{equation} 
   I_\text{FL}^\text{mode} =  \frac{1}{n-1} \sum_{i\neq k}^{n} I(\mathcal{A};G_i) 
\end{equation}


For the sake of convenience, ``mode'' is omitted when it is ``w/o SA'' in the formula. The biggest difference among the four information exchange mechanisms is the difference in 
\(\mathcal{A}.\)
In CFL w/o SA mode, the adversary node can directly obtain all the information from the nodes, i.e., 
\(\mathcal{A} = \{G_j\}_{j=1}^{n} \). Thus, 

\begin{equation}
I_\text{CFL} =  \frac{1}{n}\sum_{i=1}^{n} I(\{G_j\}_{j=1}^{n}; G_i)
\label{eq:e1}
\end{equation}
In CFL w/ SA mode, the adversary node can only obtain the average of all nodes' gradients, i.e., 
\( \mathcal{A} = \frac{1}{n} \sum_{j=1}^{n} G_j \) and thus 
\begin{equation}
I_\text{CFL}^\text{SA} = \frac{1}{n(n-1)}\sum_{k=1}^{n}( \sum_{i\neq k}^{n} I(\sum_{j=1}^{n} \frac{1}{n} G_j ; G_i | G_k))
\label{eq:e2}
\end{equation}
In DFL w/o SA mode, the adversary node can directly obtain all the information from its neighboring nodes, i.e., 
\( \mathcal{A} = \{b_{kj} G_j\}_{j=1}^{n} \).
Here, \( b_{kj} = 1 \) if the adversary node is adjacent to node \( N_j \), otherwise \( b_{kj} = 0 \).
Therefore,
\begin{equation}
I_\text{DFL} = \frac{1}{n(n-1)}\sum_{k=1}^{n}( \sum_{i\neq k}^{n} I(\{b_{kj} G_j\}_{j=1}^{n} ; G_i|G_k))
\label{eq:e3}
\end{equation}

In DFL w/ SA mode, the adversary node can only obtain the weighted values of the information from its neighboring nodes, i.e., 
\( \mathcal{A} = \sum_{j=1}^{n} a_{kj} G_j \). 
Here, \( a_{kj} \) represents the weight between node $j$ and the adversary node, and \( \sum_{j=1}^{n} a_{kj} = 1 \). Therefore, the privacy leakage is expressed as:

\begin{equation}
I_\text{DFL}^\text{SA} = \frac{1}{n(n-1)}\sum_{k=1}^{n}(\sum_{i\neq k}^{n} I(\sum_{j=1}^{n} a_{kj} G_j ; G_i | G_k))
\label{eq:e4}
\end{equation}

Without loss of generality, assume that the random variables \(\{G_i\}_{i=1}^{n}\) are independent and identically  Gaussian distributed. Our main results are given below, which shows that for the case of w/o SA, the privacy leakage of CFL is strictly no smaller than DFL. While for the case of w/ SA, it is the other way around as CFL w/ SA denotes the ideal case where the central server is trustworthy.  

\begin{prop}\label{thm.1}

\begin{equation}
I_\text{CFL} \geq I_\text{DFL} > I_\text{DFL}^\text{SA} \geq I_\text{CFL}^\text{SA},
\label{eq:compare}
\end{equation}
where the first and third equality hold if and only if the underlying graph is fully connected; the second inequality holds for all connected graphs with more than two nodes.   
\end{prop}

\begin{proof}
With \eqref{eq:e1} we have $I_\text{CFL} =I(G_i;G_i)$. Similarly, for \eqref{eq:e3} we have  $I_\text{DFL} =\frac{1}{n(n-1)}\sum_{k=1}^{n} \sum_{i\neq k}^{n} b_{ki}I(G_i;G_i)$ where $\sum_{i=1}^{n} \sum_{\substack{ k \neq i}}^{n} b_{ki} = 2m
$ as $b_{ki}=1$ if $(k,i)\in \mathcal{E}$, recall $m$ is the number of edges in the network and $m\leq \frac{n(n-1)}{2}$. Hence $I_\text{DFL} \leq  \frac{2m}{n(n-1)}I(G_i;G_i)\leq I_\text{CFL}$, which completes the proof of the first inequality, and the equality holds only if the graph is fully connected, i.e., $m=\frac{n(n-1)}{2}$. The second inequality follows simply from subadditivity of entropy \cite{lieb2002some}. We now proceed to the third inequality, given that the differential entropy of Gaussian distribution is \(h(X) = \frac{1}{2} \log (2\pi e \sigma^2)\) where $\sigma^2$ is the variance. Thus, \eqref{eq:e2} becomes 


\begin{align}
I_\text{CFL}^\text{SA}
& \overset{(a)}{=} \frac{1}{n(n-1)}\sum_{k=1}^{n}( \sum_{i\neq k}^{n} I(\sum_{j\neq k}^{n} G_j ; G_i )) \nonumber \\
& \overset{(b)}{=} \frac{1}{n(n-1)}\sum_{k=1}^{n}\sum_{j\neq k}^{n}( h(\sum_{j\neq k}^{n} G_j)-h(\sum_{j\neq k}^{n} G_j-G_i)) \nonumber \\
& \overset{(c)}{=} \frac{1}{2n(n-1)} \log (2\pi e ((\frac{n-1}{n-2})^{(n-1)n}) \nonumber\\
&=\frac{1}{2} \log (2\pi e ((\frac{n-1}{n-2})),
\label{gaussian:CFL}
\end{align}

where (a) uses the fact that $G_k$ is independent of others; (b) follows from the definition of mutual information $I(X+Y;X)=h(X+Y)-h(X+Y|X)$ and the fact that all $\{G_i\}_{i\in \mathcal{V}}$ are independent of each other; (c) uses the definition of differential entropy of Gaussian distribution. 
Thus \eqref{eq:e4} becomes

\begin{align}&I_\text{DFL}^\text{SA} 
 \overset{(a)}{=}   \frac{1}{n(n-1)}\sum_{k=1}^{n}( \sum_{i\neq k}^{n} I(\frac{1}{n-1}\sum_{j\neq k}^{n} a_{kj} G_j ; G_i )) \nonumber\\
 &\overset{(b)}{=}   \frac{1}{2n(n-1)}\sum_{k=1}^{n}( \sum_{i\neq k}^{n} \log (2\pi e \frac{\sum_{j=1}^{n}a_{kj}^2-a_{kk}^2}{\sum_{j=1}^{n}a_{kj}^2-a_{kk}^2-a_{ki}^2})) \nonumber \\
  &\overset{(c)}{=}  \frac{1}{2n(n-1)} \log 2\pi e ( \Pi_{k} \Pi_{i \neq k} ( \frac{\sum_{j=1}^{n} a_{kj}^2 - a_{kk}^2}{\sum_{j=1}^{n} a_{kj}^2 - a_{kk}^2 - a_{ki}^2} )) \nonumber \\
  &\overset{(d)}{\geq} \frac{1}{2n(n-1)} \log 2\pi e (\frac{n-1}{\sum_{{i \neq k}} {\frac{\sum_{j=1}^{n} a_{kj}^2 - a_{kk}^2 - a_{ki}^2}{\sum_{j=1}^{n} a_{kj}^2 - a_{kk}^2}}})^{(n-1)n} \nonumber\\
&= \frac{1}{2} \log 2\pi e(\frac{n-1}{ {\frac{(n-2){(\sum_{j=1}^{n} a_{kj}^2 - a_{kk}^2)}}{\sum_{j=1}^{n} a_{kj}^2 - a_{kk}^2}}}) =\frac{1}{2} \log (2\pi e(\frac{n-1}{n-2}))) \nonumber
\end{align}
where (a) and (b) follow similarly as (a) to (c) in \eqref{gaussian:CFL};  

(c) follows from product rule of logarithm; (d) follows from geometric mean - harmonic mean inequality\cite{hardy1952inequalities} where $\sqrt[n]{x_1 \cdot x_2 \cdot \dots \cdot x_n} \geq \frac{n}{\frac{1}{x_1} + \frac{1}{x_2} + \dots + \frac{1}{x_n}}$
with equality if and only if \( x_1 = x_2 = \dots = x_n \). Hence, the proof is now complete.
\end{proof}

\subsection{Validation of the main result}

To verify the findings from Proposition \ref{thm.1}, we used the NPEET library 

\cite{steeg2014npeet} to compute mutual information. We performed Monte Carlo simulations by generating 1000 samples from each random variable, modeled as Gaussian with a distribution $\mathcal{N}(0,1)$. Fig. \ref{fig:mi} compares the information leakage across four distinct scenarios, with network sizes varying from 10 to 50. We also evaluated different network topologies with varying densities, defined as $\frac{2m}{n(n-1)}$. The results are consistent with our theoretical analysis: 1) Information leakage in both DFL cases is upper-bounded by the CFL case w/o SA, and lower-bounded by the CFL case w/ SA. 2) In the DFL w/ SA scenario, information leakage decreases as the network becomes less sparse, approaching the lower bound set by the CFL w/ SA case as the network becomes fully connected. In contrast, the trend for DFL w/o SA is the opposite.    

\begin{figure}
\vspace{-8pt} 
    \centering
    \includegraphics[width=0.7\linewidth]{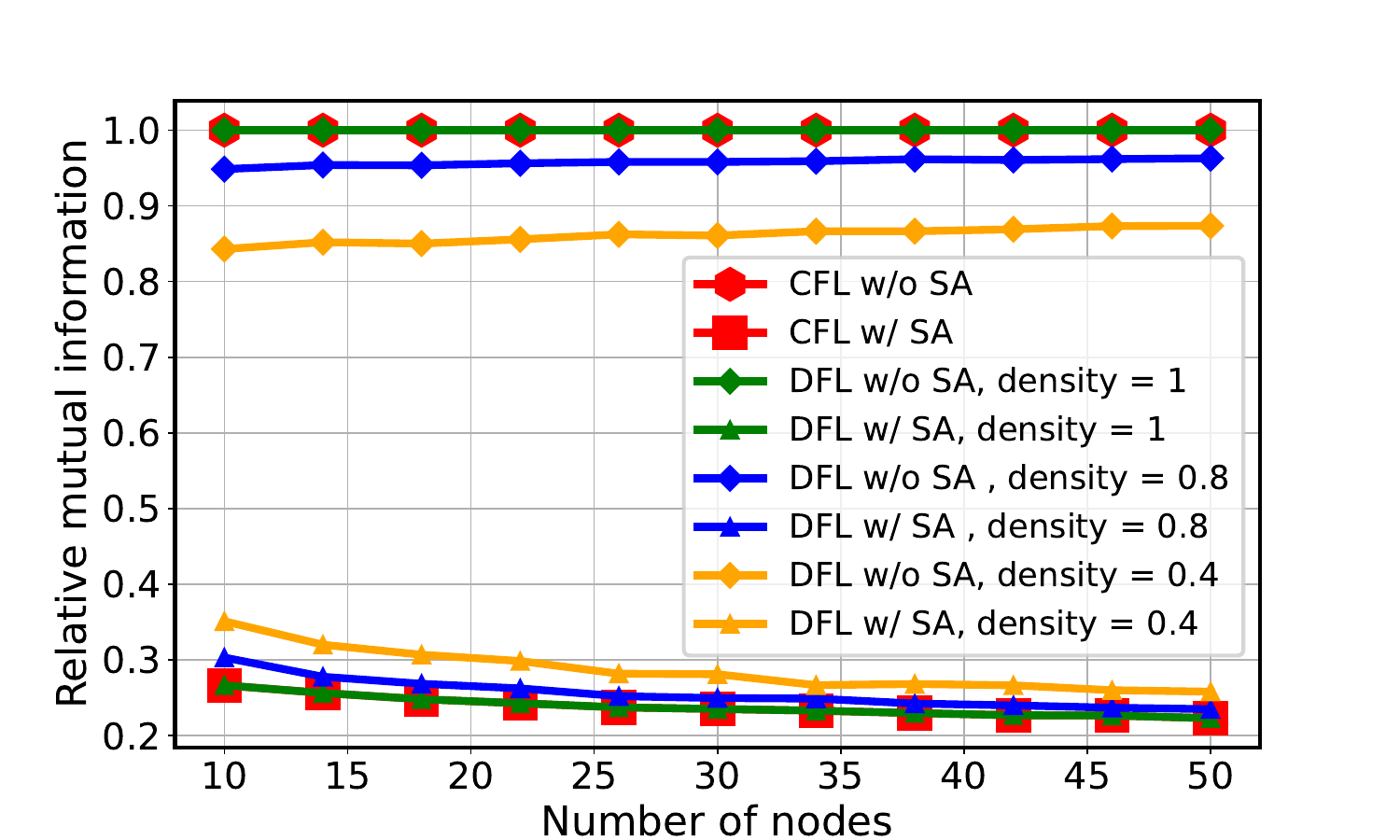}
  
    \caption{Relative mutual information ${I_{\text{FL}}^{\text{mode}}}/{I_{\text{CFL}}}$ as a function of number of nodes $n$ in the network for CFL w/o SA, CFL w/ SA, DFL w/o SA and DFL w/ SA, wherein three different network densities for DFL are considered.} 
    \label{fig:mi}
 
\end{figure}
\begin{figure*}[htbp]
    \centering
    \begin{tabular}{c c c c c}
        \begin{minipage}[b]{0.09\linewidth}  
            \centering
            \vspace{0.2cm} 
            \scriptsize{\textbf{Ground truth}}\\[0.22cm]
            \scriptsize{\textbf{CFL w/o SA}}\\[0.22cm]
            \scriptsize{\textbf{DFL w/o SA}}\\[0.22cm]
            \scriptsize{\textbf{DFL w/ SA}}\\[0.22cm]
            \scriptsize{\textbf{CFL w/ SA}}\\[0.22cm]
            \scriptsize{\textbf{}}\\[0.8cm]
        \end{minipage} &
        
        \begin{minipage}[b]{0.193\linewidth}  
            \centering
            \includegraphics[height=0.4cm]{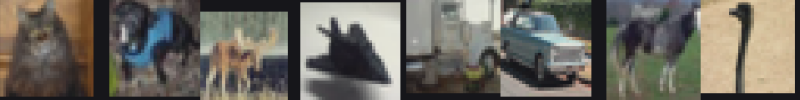}\\[0.06cm]  
            \includegraphics[height=0.40cm]{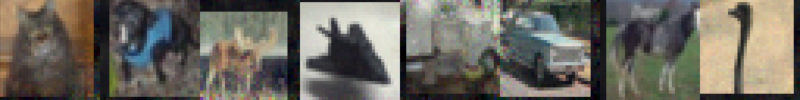}\\[0.06cm]
            \includegraphics[height=0.4cm]{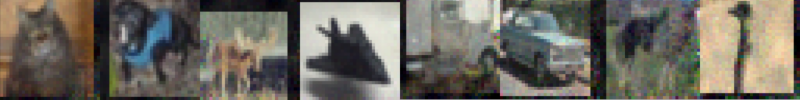}\\[0.06cm]
            \includegraphics[height=0.4cm]{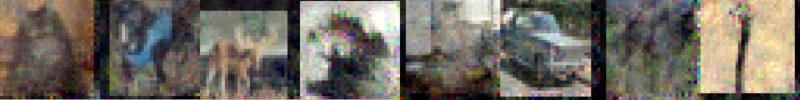}\\[0.06cm]
            \includegraphics[height=0.4cm]{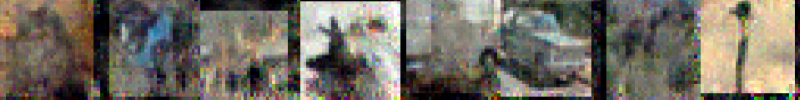}
            \text{(a) CIFAR-10}
        \end{minipage} &
        
        \begin{minipage}[b]{0.193\linewidth}  
            \centering
            \includegraphics[height=0.4cm]{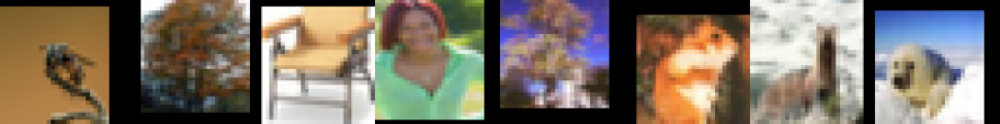}\\[0.06cm]  
            \includegraphics[height=0.4cm]{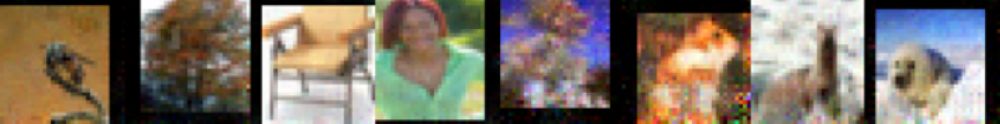}\\[0.06cm]
            \includegraphics[height=0.4cm]{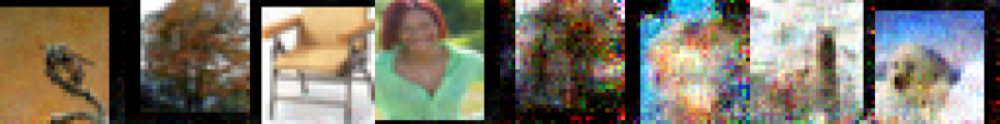}\\[0.06cm]
            \includegraphics[height=0.4cm]{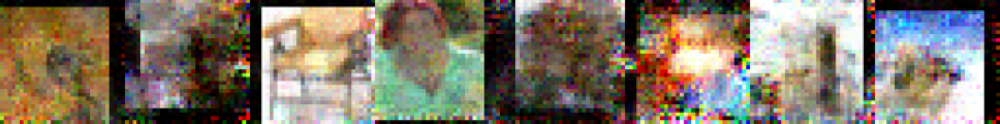}\\[0.06cm]
            \includegraphics[height=0.4cm]{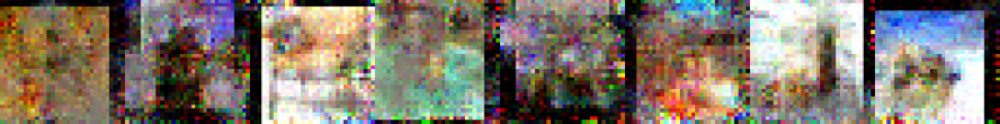}
            \text{(b) CIFAR-100}
        \end{minipage} &
        
        \begin{minipage}[b]{0.193\linewidth}  
            \centering
            \includegraphics[height=0.4cm]{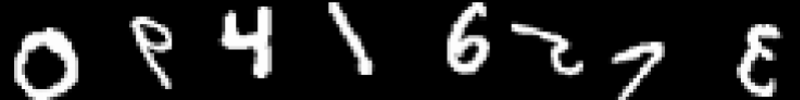}\\[0.06cm]  
            \includegraphics[height=0.4cm]{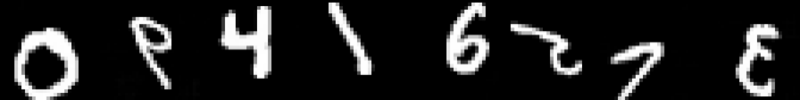}\\[0.06cm]
            \includegraphics[height=0.4cm]{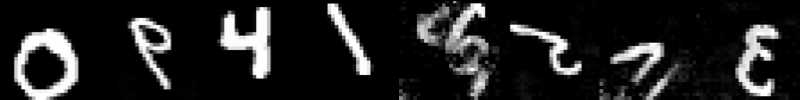}\\[0.06cm]
            \includegraphics[height=0.4cm]{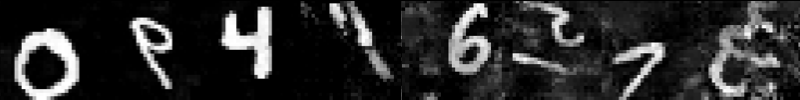}\\[0.06cm]
            \includegraphics[height=0.4cm]{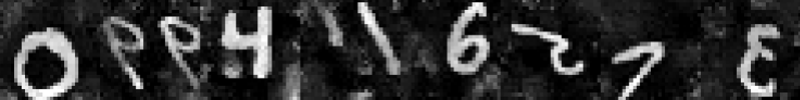}
            \text{(c) MNIST}
        \end{minipage} &

        \begin{minipage}[b]{0.19\linewidth}  
            \centering
            \includegraphics[width=1.08\linewidth]{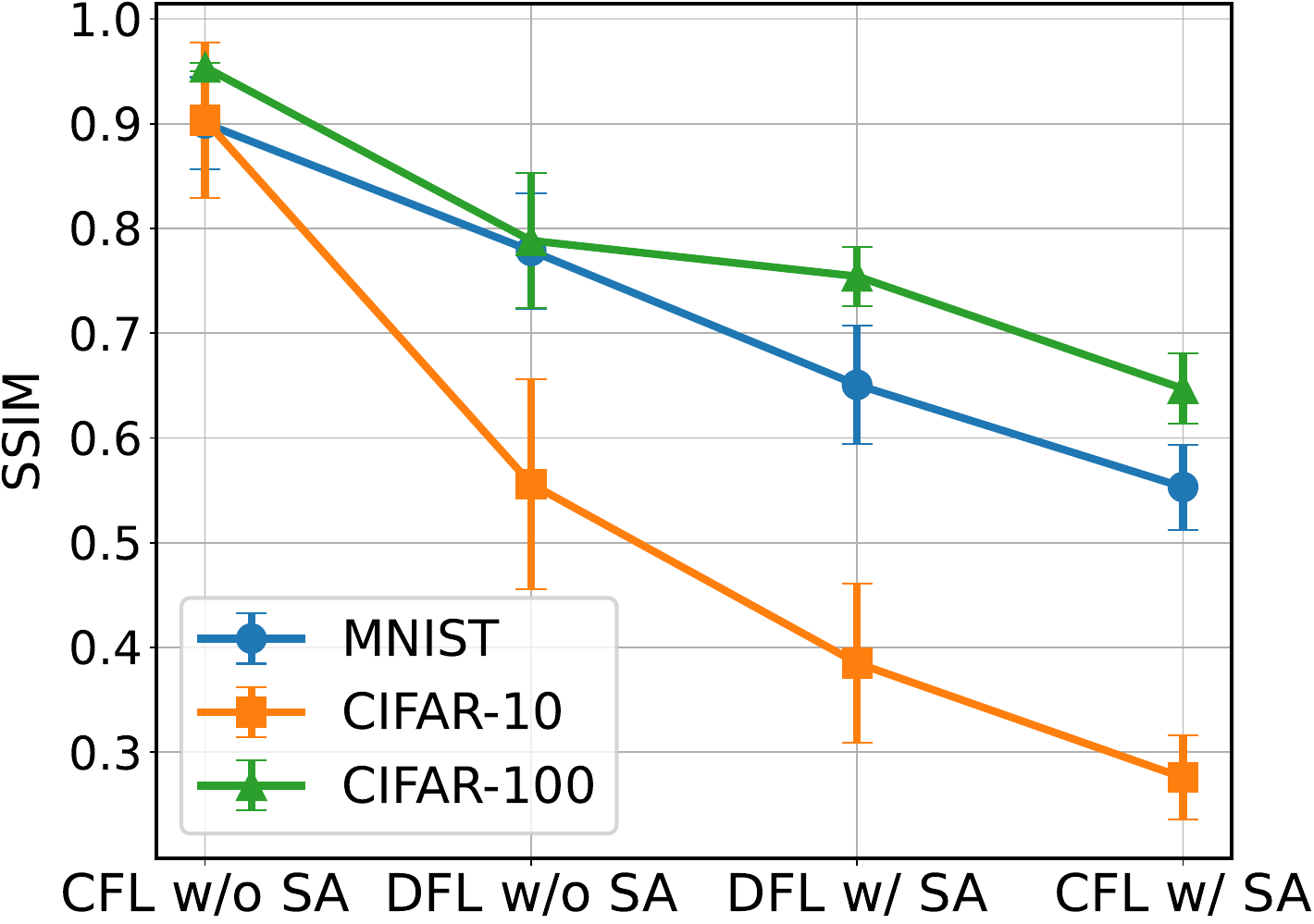}  
            \text{(d) SSIM }
        \end{minipage}
    \end{tabular}

    \caption{Samples images of ground truth and reconstructed inputs by inverting gradients on four different protocols of three datasets: (a) CIFAR-10, (b) CIFAR-100 and (c) MNIST.
    (d) Averaged SSIM values of all reconstructed inputs in each dataset for four modes.}
    \label{fig:overall_comparison}

\end{figure*}

\section{Privacy comparison via empirical attacks}
In this section, we validate our theoretical results by conducting practical privacy attacks using a gradient inversion attack \cite{geiping2020inverting},  as the theoretical analysis focuses on the gradients exchanged in the network. 

\subsection{Setup} 
For our experiments, we use the ResNet18 model with three datasets: CIFAR-10, CIFAR-100, and MNIST. We assume there are $10$ clients in the network. The learning protocol for CFL is FedSGD and for DFL is D-FedSGD. 
To assess the privacy risk, we reconstruct images by minimizing the loss over 1000 iterations. The attack learning rate is set to 0.1 for CIFAR-10 and CIFAR-100, and 1 for MNIST. Other experimental parameters remain consistent across all setups. The quality of the reconstructed inputs is measured using the Structural Similarity Index Measure (SSIM)~\cite{hore2010image}, which ranges from -1 to 1. A value of $\pm1$ indicates perfect resemblance, while 0 represents no correlation.

\subsection{Attack results}
\vspace{-1mm}
\textbf{Visualization results and comparisons}: In Fig. \ref{fig:overall_comparison} we present the reconstructed images from applying the gradient inversion attack on the four cases, using CIFAR-10, CIFAR-100, and MNIST. For DFL topology, we assume that the adversary node has 4 neighboring nodes and 5 non-neighboring nodes. For display purposes, we select images from 8 nodes: the left 4 images represent neighboring nodes, and the right 4 represent non-neighboring nodes. When applying gradient inversion attack to non-neighboring nodes, DFL w/o SA uses the average gradient of all non-neighboring nodes, and DFL w/ SA uses the average gradient of all honest nodes. In the CFL protocol, all images are treated equally.  The image quality visibly degrades from top to bottom, which is quantitatively confirmed in plot (d) with decreasing SSIM scores. Specifically, CFL w/o SA shows the highest reconstruction quality and thus poses the greatest privacy risk. In contrast, CFL w/ SA achieves the lowest SSIM, as expected, due to the presence of a trusted central server. These results support the inequality (\ref{eq:compare}) stated in Proposition \ref{thm.1}, confirming our theoretical findings.

\textbf{Impact of graph density}: In addition to the inequality results presented in Proposition \ref{thm.1}, we demonstrate that equality can only be achieved when the graph is fully connected. This observation is validated in Fig. \ref{fig:density}, where we analyze the impact of graph density on the quality of input reconstruction. See the case of DFL w/o SA as an example, when the graph topology becomes denser, the input reconstruction quality improves. This is because, with a higher density, each node has more neighbors, thereby revealing more local gradient information. When the graph is fully connected (i.e., the density equals \(1\)), the SSIM approaches that of the CFL w/o SA case, indicating that the privacy risk is at its upper bound in this scenario. In contrast, a reverse trend is observed in DFL w/ SA. As the graph density decreases, the privacy protection improves, consistent with the expected behavior and our theoretical results. This underscores the pivotal role of graph density, along with the use of SA, in influencing the degree of privacy.

\begin{figure}
\vspace{-8pt} 
    \centering
    \includegraphics[width=0.6\linewidth]{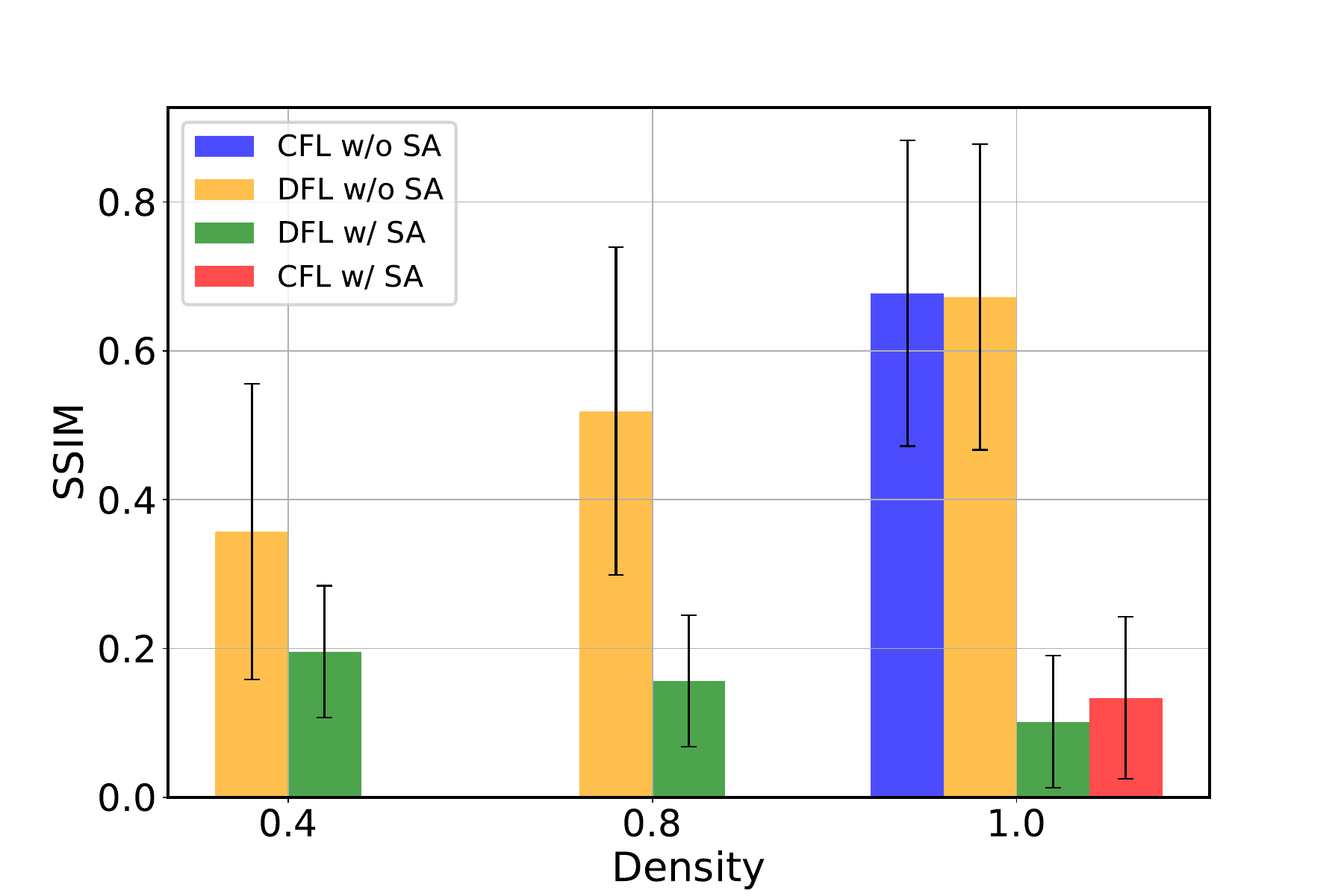}

    \caption{SSIM values of all reconstructed inputs across three network densities $0.4, 0.8$ and $1$ for DFL cases, alongside both CFL w/ and w/o SA cases.}
    \label{fig:density}
    \vspace{-8pt} 
\end{figure}

\subsection{Discussion}

\begin{figure}
    \centering
    \includegraphics[width=0.52\linewidth]{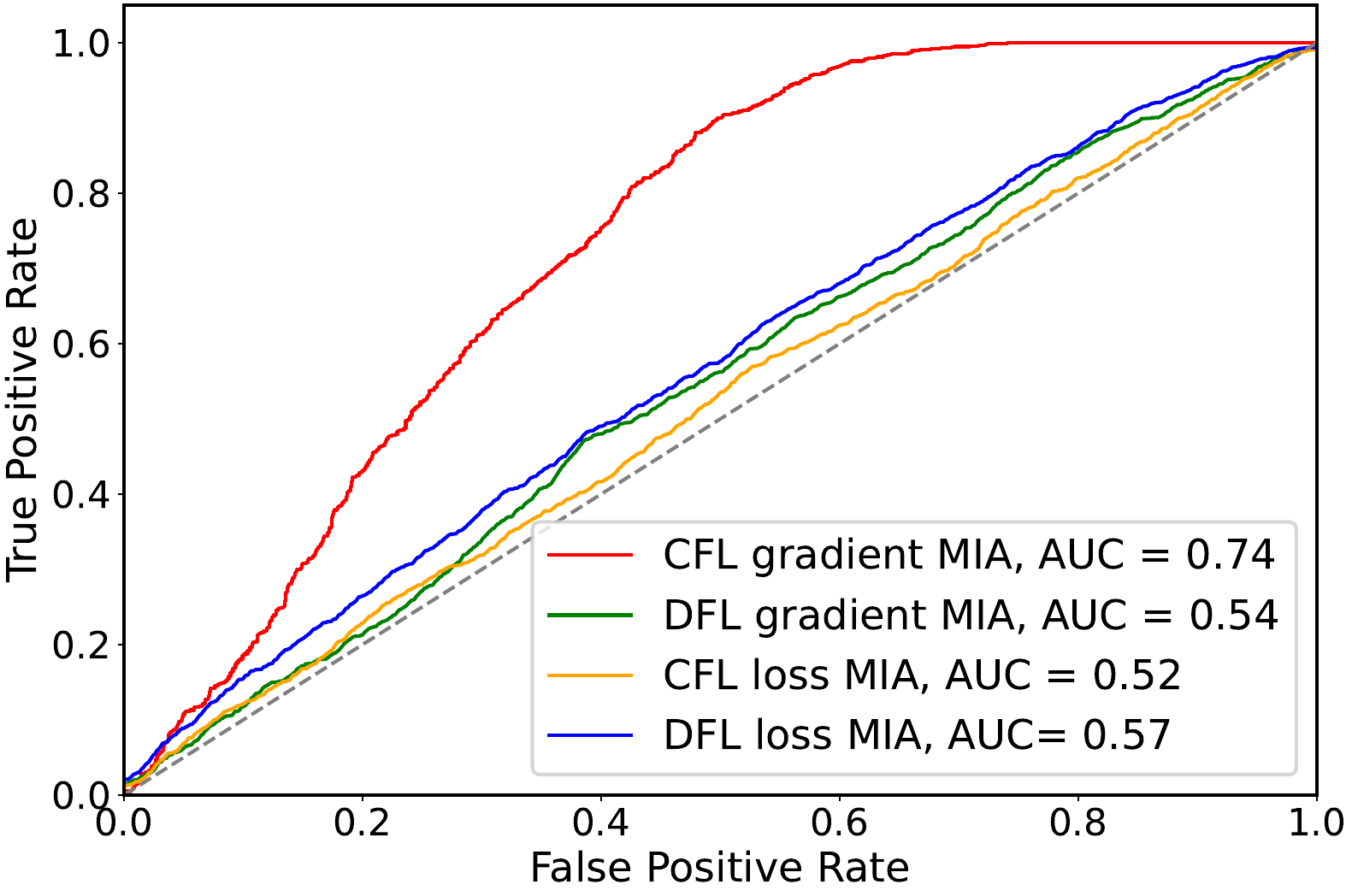}

    \caption{ROCs and attack success rates of two MIAs for CFL and DFL w/o SA cases. The AUC values represent the attack success rates.}
\vspace{-10pt} 
    \label{fig:roc}

\end{figure}

We now discuss why our conclusion contradicts the findings of previous work in \cite{Pasquini2022OnTP} from Pasquini et, al., which states that DFL offers no privacy benefits compared to CFL. First, their assumption for certain contexts does not reflect real-world setting generally, specifically, they assume a scenario in which the server in CFL remains uncompromised and implicitly the corrupt nodes in DFL have no two-hop honest neighbors \cite[Section 4]{Pasquini2022OnTP}. While this assumption is valid in certain contexts, it is quite restrictive, as one of the core motivations for DFL is the potential absence of a trusted server in practical scenarios. Second, their conclusions are based largely on empirical results without a rigorous theoretical analysis and rely on a single type of MIA  \cite{song2021systematic}. However, several studies \cite{rezaei2021difficulty, Carlini2021MembershipIA, hintersdorftrust, Li2022OnTP} have shown that this MIA tends to produce unreliable results. Therefore, using this MIA may not accurately reflect privacy leakage in FL systems.

To validate these concerns, in Fig. \ref{fig:roc} we conducted experiments with two MIAs on both CFL and DFL w/o SA. Indeed, the results are inconsistent with each other, further emphasizing the need for careful attack selection in FL. In our work, we used gradient inversion attacks, specifically tailored to FL, which align with our theoretical analysis.

\section{Conclusion}

In this paper, we compare privacy leakage between CFL and DFL from a novel perspective using information theory, specifically mutual information. We derive upper bounds for privacy leakage in both CFL and DFL and validate these theoretical results through simulation experiments. Our findings indicate that DFL generally exhibits lower privacy leakage than CFL in practical scenarios. Additionally, we perform empirical privacy attacks to further support our theoretical conclusions.  

\bibliographystyle{unsrt}  
\bibliography{references}

\end{document}